\documentclass[a4paper,11pt,reqno]{amsart}

\usepackage[margin=1in,includehead,includefoot]{geometry}
\usepackage[utf8]{inputenc}
\usepackage[T1]{fontenc}
\usepackage{lmodern,microtype}
\usepackage{hyperref}
\usepackage{booktabs}
\usepackage{mathtools,amssymb,amsthm,amsmath}
\usepackage[foot]{amsaddr}
\usepackage{graphicx}
\usepackage[table]{xcolor}
\usepackage[margin=5mm]{caption}
\usepackage{enumitem}
\usepackage[textsize=footnotesize]{todonotes}
\usepackage{mdframed}
\usepackage{xpatch}
\usepackage{dsfont}
\usepackage{nicefrac}
\usepackage{pgfplots}
\usepackage{mycommands}

\graphicspath{{./graphics/}}

\makeatletter
\let\old@setaddresses\@setaddresses
\def\@setaddresses{\bigskip{\parindent 0pt\let\scshape\relax\let\ttfamily\relax\old@setaddresses}}
\makeatother

\newtheorem{theorem}{Theorem}

\newtheorem{lemma}[theorem]{Lemma}
\theoremstyle{remark}

\hypersetup{
  pdftitle={Faster and simpler traversal of 0/1-polytopes},
  pdfauthor={Ji\v{r}\'i Fink, Petr Hlad\'ik, Arturo Merino, Ond\v{r}ej Mi\v{c}ka and Torsten M\"utze}
}

\begin{document}

\title{Faster and simpler traversal of 0/1-polytopes}

\author{Ji\v{r}\'i Fink}
\address[Ji\v{r}\'i Fink]{Department of Theoretical Computer Science and Mathematical Logic, Charles University, Prague, Czech Republic}
\email{fink@ktiml.mff.cuni.cz}

\author{Petr Hlad\'ik}
\address[Petr Hlad\'ik]{Department of Theoretical Computer Science and Mathematical Logic, Charles University, Prague, Czech Republic}
\email{hladik.petr@outlook.com}

\author{Arturo Merino}
\address[Arturo Merino]{Department of Computer Science, Universidad de Chile, Santiago de Chile, Chile}
\email{amerino@dcc.uchile.cl}

\author{Ond\v{r}ej Mi\v{c}ka}
\address[Ond\v{r}ej Mi\v{c}ka]{Department of Theoretical Computer Science and Mathematical Logic, Charles University, Prague, Czech Republic}
\email{micka@ktiml.mff.cuni.cz}

\author{Torsten M\"utze}
\address[Torsten M\"utze]{Institut f\"ur Mathematik, Universit\"at Kassel, Germany}
\email{tmuetze@mathematik.uni-kassel.de}

\thanks{Arturo Merino was supported by ANID FONDECYT Iniciaci\'on No.~11251528.
Torsten M\"utze was supported by German Science Foundation grant~522790373.}

\begin{abstract}
Recently, Merino and M\"utze [Traversing combinatorial 0/1-polytopes via optimization. SIAM J. Comput., 53(5):1257--1292, 2024] presented an algorithm for computing a Hamilton path on the skeleton of any 0/1-polytope $\conv(X)$, where $X\seq\{0,1\}^n$.
The algorithm uses as a black box an algorithm for solving the classical linear optimization problem $\min\{w\cdot x\mid x\in X\}$ for some weight vector~$w\in\mathbb{R}^n$.
The resulting delay per visited vertex on the Hamilton path is only by a $\log n$ factor larger than the time to solve one instance of the optimization algorithm.
In this paper, we make the Hamilton path algorithm simpler and faster.
Namely, we obtain an amortized delay that is only by a constant factor larger than the running time of the optimization algorithm, thus removing the $\log n$ factor.
As concrete results, this yields improved algorithms for generating bases and independent sets in a matroid, spanning trees, forests, matchings and maximum matchings in a graph, vertex covers, minimum vertex covers, independent sets and maximum independent sets in a bipartite graph, and antichains, maximum antichains and ideals in a poset.
All of these listings are Hamilton paths on the corresponding polytopes.
Furthermore, we obtain an $\cO(t_{\upright{LP}})$ amortized delay algorithm for the vertex enumeration problem on 0/1-polytopes $\{x\in\mathbb{R}^n\mid Ax\leq b\}$, where $A\in \mathbb{R}^{m\times n}$ and~$b\in\mathbb{R}^m$, and $t_{\upright{LP}}$ is the time needed to solve the linear program $\min\{w\cdot x\mid Ax\leq b\}$.
This improves upon the $\cO(t_{\upright{LP}} \log n)$ delay algorithm of Merino and M\"utze, and the previous $\cO(t_{\upright{LP}}\,n)$ delay algorithm of Bussieck and L\"ubbecke from~1998.
\end{abstract}

\maketitle

\section{Introduction}

Given a set~$X$ of 0/1-vectors of length~$n$, its convex hull $P\coloneq \conv(X)$ is a \defi{0/1-polytope} in~$\mathbb{R}^n$.
Classical and heavily studied examples are the spanning tree polytope of a graph, or more generally, the base polytope of a matroid~\cite{MR510371}, the matching polytope or the independent set polytope of a graph~\cite{MR371732} (see Figure~\ref{fig:ssp}~(a)), and the chain and order polytopes of a poset~\cite{MR824105}.
In each of these cases, we take as the set~$X$ the \defi{indicator vectors} of the corresponding combinatorial objects.
For example, given a spanning tree~$T$ of a graph~$G$ with edges labeled~$1,\ldots,n$, we write $\indicator_{T}\in\{0,1\}^n$ for the vector whose $i$th entry is~1 if the edge~$i$ is contained in the spanning tree~$T$ and~0 otherwise.
Then the spanning tree polytope is defined as $\conv(X)$ for $X\coloneq \{\indicator_{T}\mid T \text{ spanning tree of }G\}$.

0/1-polytopes arising from combinatorial structures such as matroids, graphs and posets have been studied heavily, most prominently in \defi{combinatorial optimization}, where the goal is to find the best vector~$x\in X$ subject to some objective function.
A prototypical task is to solve the \defi{linear optimization problem} $\min\{w\cdot x\mid x\in X\}$, where $w\in\mathbb{R}^n$ are some fixed real weights.
A tremendous amount of previous work has been devoted to finding the best algorithms to solve this problem in the different cases mentioned above where~$X$ encodes a concrete set of combinatorial objects.
This includes the classical problems of minimum spanning tree and minimum weight matching (usually equivalently stated as maximum weight matching).

In addition to being useful for optimization, combinatorial polytopes often encode interesting information about the underlying combinatorial structures.
For example, the volume of the chain polytope and order polytope of a poset are proportional to the number of linear extensions of the poset~\cite{MR824105}, so an enumeration problem can be expressed as a volume computation.
Furthermore, the adjacency relation on the polytope, i.e., which pairs of vertices are connected by an edge, defines a natural closeness relation on the combinatorial objects.
This is captured by the \defi{skeleton} of the polytope, which is the graph formed by its vertices and edges (=0- and 1-dimensional faces).
For example, on the spanning tree polytope, edges of the polytope connect pairs of spanning trees that differ in an exchange of a single edge~\cite{MR510371}.
Similarly, on the perfect matching polytope, edges of the polytope connect pairs of matchings that differ in an alternating cycle~\cite{MR371732}.
Furthermore, on the independent set polytope, edges of the polytope connect pairs of independent sets whose symmetric difference is connected~\cite{MR371732}; see Figure~\ref{fig:ssp}~(a).

\begin{figure}[t!]
\includegraphics[page=1]{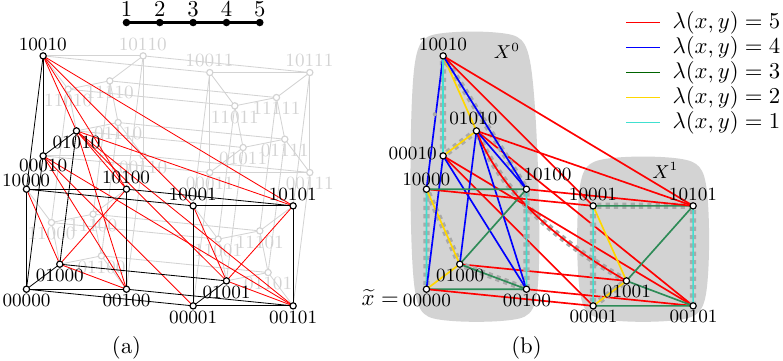}
\caption{(a) Independent set polytope of a 5-vertex path, with vertices numbered~$1,\ldots,5$ from left to right.
Vertices are the indicator vectors of all independent sets of the path, i.e., sets of vertices that do not span any edges.
They correspond to bitstrings of length~5 with no two adjacent 1s.
Distance-1 edges are in black, and the others in red.
(b)~Hamilton path computed by Algorithm~P\ppp{} on the independent set polytope from the left, with initial vertex~$\tx=00000$.
Edges~$\{x,y\}$ are colored according to the values~$\lambda(x,y)$, and the algorithm greedily walks along edges with smallest possible $\lambda$-value that leads to new vertex.
The shadings highlight the two faces of the polytope in which the last bit is either~0 or~1, and the algorithm transitions between them exactly once.
The example in this figure will be used repeatedly throughout this paper.
}
\label{fig:ssp}
\end{figure}

In this paper we focus on the problem of computing a \defi{Hamilton path} on the skeleton of a 0/1-polytope~$P=\conv(X)$, i.e., a path on the skeleton that visits every vertex exactly once.
It has long been known that such a Hamilton path always exists, for any 0/1-polytope~\cite{MR762893}.
In other words, our goal is to list all combinatorial objects encoded by~$X$, each object exactly once, such that any two consecutive objects are connected by an edge of the polytope.
It is noteworthy here that the (often very large) set~{\bf $X$ is not given explicitly as input} to the algorithm, but rather implicitly, otherwise there would be no point in computing it.
E.g., given a graph as input, the task is to compute a Hamilton path on the spanning tree polytope of the graph, or, given a poset as input, the task is to compute a Hamilton path on the order polytope of the poset.
In the different applications, this yields listings of combinatorial objects in which consecutive objects differ only by small changes, i.e., we obtain \defi{combinatorial Gray codes}~\cite{MR1491049,MR4649606}.
These are part of the arsenal of \defi{combinatorial generation}, where the goal is to list all elements of~$X$, each exactly once, as efficiently as possible, which is covered in depth in Vol.~4A of Knuth's seminal series `The Art of Computer Programming'~\cite{MR3444818}.
The running time of a generation algorithm is usually measured by its \defi{delay}, i.e., the time spent between visiting any two consecutive objects, and before the first and after the last visit.
The delay can be specified as a worst-case guarantee, which is the maximum of the aforementioned quantities, or as an amortized guarantee, which is the average.
The amortized delay is simply the total time of the algorithm to compute the entire listing, divided by the total number of objects generated.

\subsection{The Merino--M\"utze algorithm}

In a recent paper, Merino and M\"utze~\cite{MR4795009} presented a novel algorithm for computing a Hamilton path on the skeleton of any 0/1-polytope~$P=\conv(X)$, $X\seq\{0,1\}^n$, which has a number of striking features, listed in Section~1.4 of their paper.
Most importantly, their algorithm uses as a black box an algorithm for solving instances of the linear optimization problem over~$X$, called repeatedly with different weight vectors~$w\in\mathbb{R}^n$.
The resulting worst-case delay of their Hamilton path algorithm is then $\log n$ times the running time to solve one instance of the linear optimization problem.
In other words, their result reduces the combinatorial generation problem---in the stronger form of computing a Hamilton path on the associated polytope---to the combinatorial optimization problem on the same polytope, albeit at the cost of an extra $\log n$ factor in the (worst-case) delay, which originates from a binary search.
Their algorithm thus embodies a very versatile general framework to derive a generation algorithm for a set of combinatorial objects by exploiting a state-of-the-art optimization algorithm for the corresponding objects.
Applying this reduction with the best-known optimization algorithms from the literature yields a large number of generation algorithms for concrete combinatorial objects, which in many cases are the best known.
\cite[Table~1]{MR4795009} lists 26 examples explicitly, including bases and independent sets in a matroid, spanning trees, forests, matchings and maximum matchings in a graph, vertex covers, minimum vertex covers, independent sets and maximum independent sets in a bipartite graph, and antichains, maximum antichains and ideals in a poset.

Furthermore, one of their applications is the \defi{vertex enumeration problem}~\cite{MR0060202,MR1174359,MR1785299}, which is fundamental in discrete and computational geometry.
In this problem, we are given a linear system of inequalities~$Ax\leq b$, where $A\in\mathbb{R}^{m\times n}$ and~$b\in\mathbb{R}^m$, and the problem is to compute all vertices of the polytope~$P\coloneq \{x\in\mathbb{R}^n\mid Ax\leq b\}$.
In other words, given the half-space representation of~$P$, we want to compute the vertex-representation of~$P$.
For the special case where $P$ is a 0/1-polytope, i.e., all vertex coordinates are from~$\{0,1\}$, Bussieck and L\"ubbecke~\cite{MR1659922} described an algorithm for generating the vertices of~$P$ with delay~$\cO(t_{\upright{LP}}\,n)$, where $t_{\upright{LP}}$ is the time needed to solve the linear program (LP) $\min\{w\cdot x\mid Ax\leq b\}$ for some weight vector~$w\in\mathbb{R}^n$.
The algorithm of Merino and M\"utze improves upon this, reducing the delay from~$\cO(t_{\upright{LP}}\,n)$ to~$\cO(t_{\upright{LP}}\,\log n)$.
The problem statement does not require that the vertices of the polytope are listed in the order of a Hamilton path, but the algorithm of Merino and M\"utze guarantees this as an additional bonus (on top of being faster).

\subsection{Our results}

Our main contribution is to make the algorithm of Merino and M\"utze simpler and faster.
Specifically, the new algorithm to compute a Hamilton path on the skeleton of any 0/1-polytope has an amortized delay that is only a constant times the running time to solve one instance of the linear optimization problem (see Algorithm~P\ppp{} and Theorem~\ref{thm:main} below).
The constant is relatively small, namely~4.
Informally, this shows that combinatorial generation is (up to constant factors) not harder than combinatorial optimization.
This improvement is achieved by slightly modifying one of the auxiliary optimization problems, such that each solution reveals more information to the Hamilton path algorithm.
Maintaining this information in a stack enables us to get rid of the binary search needed in the original reduction, thus removing the $\log n$ factor in the delay.
This improves the running times for each of the 26 generation problems stated in~\cite[Table~1]{MR4795009} by a factor of~$\log n$.
In particular, we obtain an algorithm for the vertex enumeration problem on 0/1-polytopes with amortized delay~$\cO(t_{\upright{LP}})$, further improving on the aforementioned result of Bussieck and L\"ubbecke.
Notably, our new algorithm retains all the other features advertised in Section~1.4 of~\cite{MR4795009}, in particular, it is very simple and easy to implement with few lines of code, and it performs neither polyhedral nor numerical computations, but only purely combinatorial operations.

In our new result the constant factor bound on the delay holds in an amortized sense, whereas in the previous result the $\log n$ factor bound on the delay holds in the worst case.
This distinction is mostly a technical artefact, as standard buffering techniques could be applied to turn amortized delay into worst-case delay, though at the cost increasing the total running time and using additional memory (cf.~\cite{MR3627876,MR4075363,MR4587202}).

\subsection{Outline of this paper}

In Section~\ref{sec:algo} we present the new Hamilton path algorithm and prove its correctness (Theorem~\ref{thm:algo}).
In Section~\ref{sec:opt} we establish the link to combinatorial optimization, thus proving the amortized delay bounds for the algorithm (Theorem~\ref{thm:main}).
In Section~\ref{sec:exp} we present a computational study to compare the performance of the Bussieck--L\"ubbecke algorithm, the Merino--M\"utze algorithm and our new algorithm when solving the vertex enumeration problem on a variety of combinatorial 0/1-polytopes.

\section{The algorithm}
\label{sec:algo}

In this section we present Algorithm~P\ppp{} to compute a Hamilton path on the skeleton of any 0/1-polytope.
The underlying principles of the algorithm and the Hamilton path it computes are identical to Algorithm~P\sss{} presented in~\cite{MR4795009}, and we elaborate on the differences and resulting improvements later.
However, in order to keep the paper self-contained, we provide an independent and streamlined proof here.

\subsection{Preliminaries}

For integers $i$ and~$j$ we define the \defi{interval} $[i,j]\coloneq \{i,i+1,\ldots,j\}$.
If $i>j$ we have~$[i,j]=\emptyset$.
We also define~$[n]\coloneq [1,n]=\{1,\ldots,n\}$.

Throughout this paper, $X\seq\{0,1\}^n$ denotes the set of vertices of a 0/1-polytope~$P\coloneq \conv(X)$ on which we aim to compute a Hamilton path.
Often, we think of~$X$ as the set of indicator vectors encoding a set of combinatorial objects.
For a predicate~$P$ and a real-valued function~$f$ defined on elements in~$X$ that satisfy the predicate~$P$, we write
\[ \argmin[P(x)\mid f(x)]\coloneq \argmin_{P(x)} f(x) \]
for the minimizers of~$f$ on this subset of~$X$.
For an element~$x\in X$ we define $X-x\coloneq X\setminus\{x\}$.

We write $\varepsilon$ for the empty string.
For a bitstring $x=(x_1,\ldots,x_n)\in\{0,1\}^n$ and integers $1\leq i\leq j\leq n$ we write $x_{[i,j]}\coloneq (x_i,x_{i+1},\ldots,x_j)$ for the substring of~$x$ from position~$i$ to~$j$.
Furthermore, we define
\[ S(x)\coloneq (\varepsilon,x_{[n,n]},x_{[n-1,n]},\ldots,x_{[2,n]},x_{[1,n]}) \]
as the sequence of all suffixes of~$x$, ordered by increasing lengths, starting with length~0 (which is the empty string) and ending with length~$n$ (which is the entire string).
Given two bitstrings~$x,y\in\{0,1\}^n$, we write $d(x,y)$ for their \defi{Hamming distance}, i.e., the number of bits in which they differ.
Furthermore, for two distinct bitstrings~$x,y\in\{0,1\}^n$ we define $\lambda(x,y)\coloneq \max\{i\in[n]\mid x_i\neq y_i\}$, which is the rightmost (=maximum) index in which they differ.

For a bitstring~$x=(x_1,\ldots,x_n)$ we write~$x^-\coloneq (x_1,\ldots,x_{n-1})$ for the string obtained by deleting the last bit.
For a set~$X\seq\{0,1\}^n$ and $b\in\{0,1\}$ we define $X^b\coloneq \{x\in X\mid x_n=b\}$.
We also define $X^-\coloneq\{x^-\mid x\in X\}$ for the set of strings obtained by deleting the last bit in each of them.
We also write $X^{b-}\coloneq (X^b)^-$.

A \defi{listing}~$L=x_1,\ldots,x_\ell$ of~$X\seq\{0,1\}^n$ is a total ordering of the bitstrings in~$X$, i.e., each element appears exactly once in~$L$.
For a bit~$b\in\{0,1\}$ we write $L^b$ for the subsequence of~$L$ given by all bitstrings whose last bit equals~$b$.
We also define $L^-\coloneq x_1^-,\ldots,x_\ell^-$, i.e., this is the listing of bitstrings of length~$n-1$ obtained by deleting the last bit of each entry of~$L$.
Furthermore, we write $L^{b-}\coloneq (L^b)^-$.

\subsection{Genlex listings and suffix trees}

A listing~$L$ of~$X\seq\{0,1\}^n$ is \defi{genlex}, if all bitstrings with the same suffix appear consecutively in~$L$.
Equivalently, this means that $L$ either has the form $L=L^0,L^1$ or $L=L^1,L^0$ and both~$L^{0-}$ and~$L^{1-}$ are genlex listings.
In words, all bitstrings ending with~0 appear before all bitstrings ending with~1, or vice versa, and this property holds recursively within each block in which the same last bit is removed.
Clearly, this is a generalization of colexicographic order, in which only the form~$L=L^0,L^1$ is allowed, i.e., 0 always appears before~1.
Unlike colexicographic order, which is unique for a set~$X$, there are in general many different genlex orderings for a given set~$X$.

\begin{figure}[t!]
\begin{tabular}{ccc}
\includegraphics[page=1]{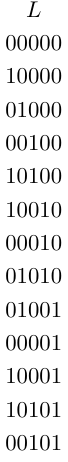} &
\includegraphics[page=2]{tree2} &
\includegraphics[page=4]{tree2} \\
(a) & (b) & (c)
\end{tabular}
\caption{(a)~The genlex listing~$L$ of independent sets of the 5-vertex path from Figure~\ref{fig:ssp}.
(b)~The binary tree structure in~$L$.
(c)~The corresponding suffix tree~$\cT(L)$.
One of its root-to-leaf paths is highlighted, together with the corresponding branchings and unseen branchings.}
\label{fig:suffix-tree}
\end{figure}

The key feature of our algorithm is that it computes a Hamilton path on a 0/1-polytope $\conv(X)$ in a genlex order.
In polytope language, this means that the algorithm first visits all vertices on the face~$\conv(X^0)$ and then all vertices on the face~$\conv(X^1)$, or vice versa, depending on the choice of starting vertex, with only a single transition of the last bit; see Figure~\ref{fig:ssp}~(b).
Clearly, each of these two faces is itself a lower-dimensional 0/1-polytope, so the same recursive strategy applies to it.

Given~$X\seq\{0,1\}^n$, we define the \defi{suffix tree of~$X$}, denoted~$\cT(X)$, as an unordered rooted tree whose vertices in depth~$k$ are all possible suffixes of length~$k$ appearing in~$X$, and the parent-child relation is the suffix relation.
In particular, the root of~$\cT(X)$ is the empty string~$\varepsilon$, and the set of leaves of~$\cT(X)$ equals~$X$.
Equivalently, $\cT(X)$ has the empty string~$\varepsilon$ as root, and its children are the suffix trees obtained from~$\cT(X^{0-})$ and~$\cT(X^{1-})$ by appending a 0-bit or 1-bit to each of their vertices, respectively, provided that the sets~$X^0$ and~$X^1$ are non-empty (otherwise the corresponding subtrees are missing).
Given any genlex listing~$L$ of~$X$, the \defi{suffix tree of~$L$}, denoted $\cT(L)$, is an ordered rooted tree obtained from~$\cT(X)$ by ordering subtrees such that the leaves of the tree appear exactly in the order~$L$; see Figure~\ref{fig:suffix-tree}.
We observe that~$\cT(L)$ is obtained from~$\cT(X)$ by choosing, at every vertex with two children, one of two possible orderings of subtrees.
Also note that for any~$x\in X$, the sequence~$S(x)$ of all suffixes of~$x$ defines the unique root-to-leaf path from~$\varepsilon$ to~$x$ in the tree~$\cT(X)$.

Given a bitstring~$x\in X$, we define its \defi{set of branchings} as $B(x)\coloneq \{\lambda(x,y)\mid y\in X-x\}$.
In words, this is the set of prefix lengths of~$x$ that can be modified to reach another element of~$X$ (different from~$x$).
In the suffix tree~$\cT(X)$ we have the following:
The integer~$\lambda(x,y)$ is the distance from the leaves~$x$ and~$y$ at which the root-to-leaf paths~$S(x)$ and~$S(y)$ split.
Thus, $B(x)$ contains the distances from~$x$ to all vertices on the root-to-leaf path~$S(x)$ that have two children, one of them lying on the path and the other one not.

Given a genlex listing~$L'=x_1,\ldots,x_i$ of a subset~$X'\seq X$ (not necessarily all of them), the set of \defi{unseen branchings} is defined as $B(L')=B(x_1,\ldots,x_i)\coloneq \{\lambda(x_i,y)\mid y\in X\setminus \{x_1,\ldots,x_i\}\}$; see Figure~\ref{fig:suffix-tree}~(c).
In words, this is the set of prefix lengths of~$x_i$ that can be modified to reach an element of~$X$ that is not part of the listing~$L'$.
Formally, for any genlex listing~$L$ of~$X$ that extends~$L'$, the suffix tree~$\cT(L')$ is a subtree of~$\cT(L)$ to the left of the root-to-leaf path~$S(x_i)$ (which is the rightmost path in~$\cT(L')$), i.e., all such suffix trees~$\cT(L)$ differ only in swapping subtrees to the right of the path~$S(x_i)$, whereas everything to the left of the path is fixed.

\subsection{The algorithm}

The idea of the algorithm to compute a Hamilton path on the skeleton of a 0/1-polytope~$\conv(X)$ is to traverse the (unordered) suffix tree~$\cT(X)$ by a depth-first search, thus turning it into an (ordered) suffix tree~$\cT(L)$ step by step, where $L$ is the resulting genlex Hamilton path.
Given the partial Hamilton path~$x_1,\ldots,x_i$ computed up to some point, the algorithm greedily takes the smallest unseen branching in~$B(x_1,\ldots,x_i)$ to move from the current root-to-leaf path~$S(x_i)$ to the next vertex~$x_{i+1}$ and its corresponding path~$S(x_{i+1})$.
This corresponds to changing~$x_i$ in the shortest possible prefix to reach a new vertex~$x_{i+1}$.

Consider now the pseudocode of the Algorithm~P\ppp{} and the auxiliary function~$\branching()$ stated below.
The initial vertex~$\tx\in X$ of the Hamilton path is provided as input to the algorithm.
We first explain the main algorithm, whereas the auxiliary function is discussed subsequently.
The algorithm consists of an initialization step~P1, and a loop spanning the remaining steps~P2--P6.
The variable~$x$ stores the current vertex on the Hamilton path, and it is initialized in step~P1.
In step~P2, the current vertex~$x$ is visited.
In steps~P3--P5, the next vertex~$y$ on the Hamilton path is computed, and then assigned to~$x$ (step~P5).
After some data structure updates in step~P6, the loop goes back to step~P2.

The main data structure used by the algorithm is a stack~$U$ that contains all unseen branchings of the current vertex~$x$ in decreasing order.
Specifically, if $x_1,\ldots,x_i$ are the vertices visited in the first $i$ iterations of the main loop, then at the beginning of the $i$th iteration the stack~$U$ contains all elements of~$B(x_1,\ldots,x_i)$, sorted decreasingly, i.e., the smallest element at the top of the stack and the largest element at the bottom.
In particular, if~$U$ is empty, then there are no unseen branchings, so all of~$X$ has been visited, causing the algorithm to terminate (step~P3).

We now discuss steps~P3--P5 that compute a neighbor~$y$ of~$x$ on the skeleton of~$\conv(X)$.
In step~P3, the smallest unseen branching~$\beta$ is popped from the stack~$U$.
This is the bit position that will be flipped in~$x$ (plus possibly some earlier bits) to reach a vertex~$y\in X-x$ that is new, i.e., that is not among the vertices visited so far, by the definition of unseen branchings.

\begin{algo}{Algorithm~P\ppp{}}{History-free traversal of 0/1-polytope by shortest prefix changes}
For a set $X\seq \{0,1\}^n$, this algorithm greedily computes a Hamilton path on the skeleton of the 0/1-polytope $\conv(X)$, starting from an initial vertex~$\tx$.
\begin{enumerate}[label={\bfseries P\arabic*.}, leftmargin=8mm, noitemsep, topsep=3pt plus 3pt]
\item{} [Initialize] Set $x \gets \tx$ and call $\branching([1,n])$.
\item{} [Visit] Visit~$x$.
\item{} [Min.\ unseen branching] Terminate if $U$ is empty.
Otherwise set $\beta\!\gets\! U.\pop()$.
\item{} [Closest vertices] Compute the set~$N$ of vertices~$y$ with $\lambda(x,y)=\beta$ of minimum Hamming distance from~$x$, i.e., $N\gets \argmin[y\in X-x \,\wedge\, \lambda(x,y)=\beta \mid d(x,y)]$.
\item{} [Tiebreaker+update~$x$] Pick an arbitrary vertex $y \in N$ and set $x\gets y$.
\item{} [Update~$U$] Call $\branching([1,\beta-1])$ and goto~P2.
\end{enumerate}
\end{algo}

\begin{algo}{Function~$\branching(I)$}{Update stack}
Computes the set~$B(x)\cap I$ and pushes the elements onto the stack~$U$ in decreasing order.
\begin{enumerate}[label={\bfseries B\arabic*.}, leftmargin=8mm, noitemsep, topsep=3pt plus 3pt]
\item{} [Pick one branching] Return if $X(x,I)=\emptyset$.
Otherwise pick an arbitrary vertex $y\in X(x,I)$.
\item{} [Recurse] Set $\beta\leftarrow \lambda(x,y)$ and $[\alpha,\gamma]\leftarrow I$.
Then call $\branching([\beta+1,\gamma])$, $U.\push(\beta)$, and call $\branching([\alpha,\beta-1])$.
\end{enumerate}
\end{algo}
We emphasize that all such vertices $y\in X-x$ with $\lambda(x,y)=\beta$ are new, and in step~P4 the algorithm computes the subset~$N$ of them with minimum Hamming distance from~$x$.
This is the one and only place in the algorithm where geometry in~$\mathbb{R}^n$ enters the picture.
Namely, an easy and well-known sufficient condition from polytope theory (see \cite[Lemma~14]{MR4795009} and \cite[Proposition~2.3]{MR762893}) guarantees that a vertex~$y$ with minimum Hamming distance from~$x$ has the property that~$\{x,y\}$ is indeed an edge of the polytope~$\conv(X)$.
There can be several vertices with minimum Hamming distance from~$x$, i.e., $N$ may contain more than one element.
Any one of them is new and connected to~$x$ by an edge, and thus eligible to become the next vertex~$y\in N$, i.e., ties can be broken arbitrarily (step~P5).
In practice steps~P4+P5 will typically be combined to directly pick an arbitrary vertex~$y\in N$, without computing the entire set~$N$, but we find it conceptually cleaner to keep the two steps separated, emphasizing the fact that ties may arise.

It remains to discuss the auxiliary function~$\branching()$ called in step~P6, and also for initialization in step~P1.
Its purpose is to update the stack~$U$ of unseen branchings.
Specifically, the call~$\branching(I)$ for a subinterval~$I$ of~$[n]=[1,n]$ computes the set~$B(x)\cap I$ (where $x$ is the vertex to be visited in the next visit step~P2, i.e., after the update $x\leftarrow y$ in step~P5) and pushes the elements of this set onto the stack in decreasing order (smallest element on top).
Note that $B(x)\cap I$ is the set of \emph{all} branchings of~$x$ inside the interval~$I$, not just the unseen ones, but we shall see momentarily that these are in fact the same for the intervals~$I$ for which the function is called.
Before we explain how the function~$\branching()$ works internally, let us explain how it is being called from~Algorithm~P\ppp{} in order to update the stack~$U$ correctly.

\begin{figure}[b!]
\includegraphics[page=5,scale=0.78]{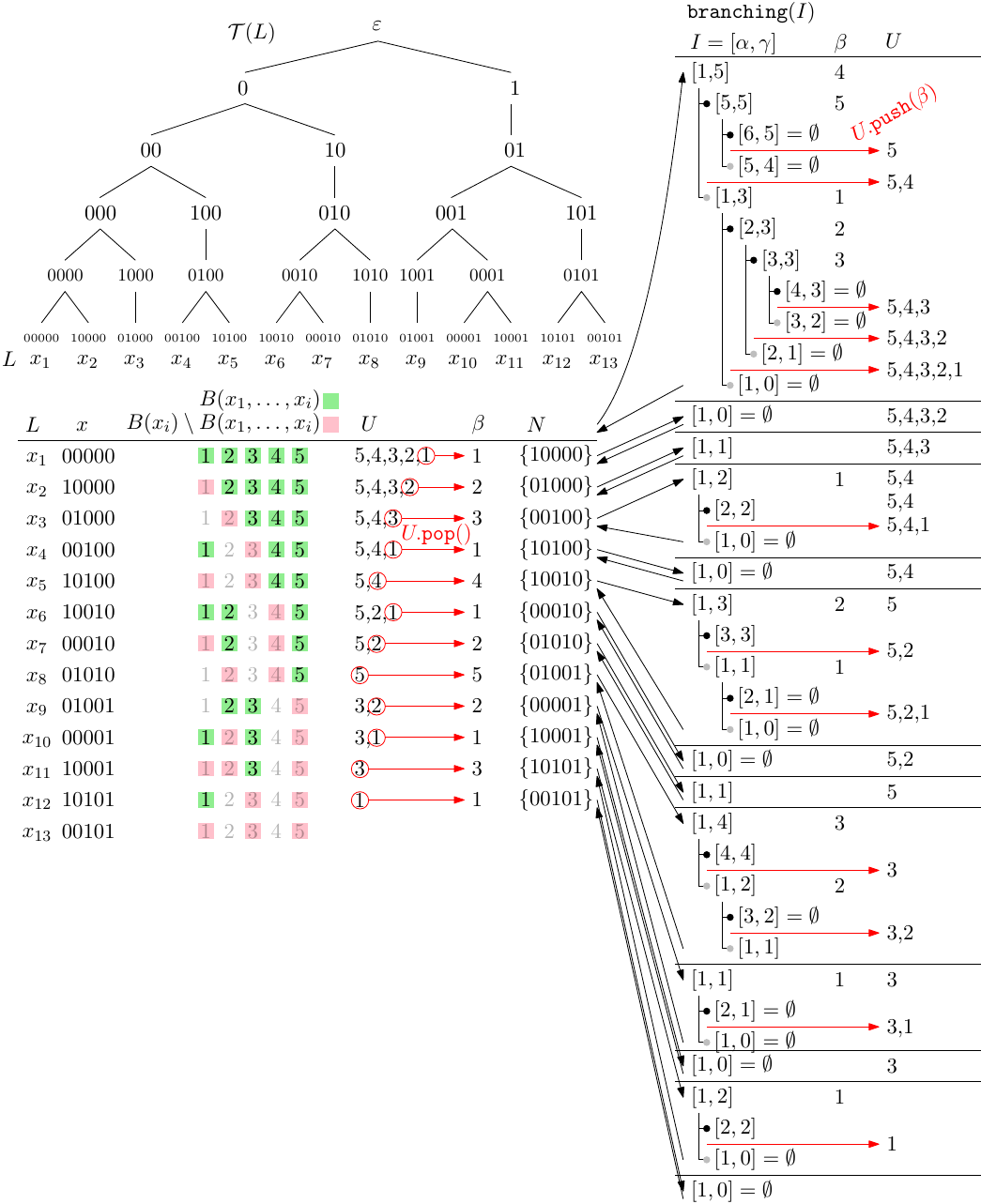}
\caption{Protocol of Algorithm~P\ppp{} when computing the Hamilton path from Figure~\ref{fig:ssp}~(b).
Each line shows the state of the variables~$x,U$ in step~P3, and the values of~$\beta$ and~$N$ as computed in steps~P4 and~P5, respectively.
The protocols of the calls to the function~$\branching()$ are shown at the right hand side.
Whenever such a call incurs two recursive calls, the first one is indicated with a black bullet, the second one by a gray bullet, and the push operation to the stack in between is highlighted by an arrow.
}
\label{fig:protocol}
\end{figure}

At the beginning, when no vertex has been visited yet, all branchings are unseen, which is why the call~$\branching([1,n])$ in the initialization step~P1 is correct.
Now consider the $i$th iteration of the main loop, and let $x_1,\ldots,x_i$ be the vertices visited in the first $i$ iterations.
The variable~$x$ contains the value~$x_i$, and is assigned some~$y\in N$ in step~P5, which becomes the next vertex~$x_{i+1}$ to be visited.
Our invariant is that at the beginning of the $i$th iteration~$U$ contains all elements of~$B(x_1,\ldots,x_i)$ in decreasing order.
The smallest unseen branching~$\beta$ is popped from the stack~$U$ in step~P3, reflecting the fact that $\beta\notin B(x_1,\ldots,x_i,x_{i+1})$, namely, the branching~$\beta$ is not unseen anymore.
We note that the root-to-leaf paths~$S(x_i)$ and~$S(x_{i+1})$ split at distance~$\beta$ from the leaves in the suffix tree~$\cT(X)$, and they coincide on all vertices closer to the root.
Consequently, we have
\[ B(x_1,\ldots,x_i,x_{i+1})\cap [\beta+1,n]=B(x_1,\ldots,x_i)\cap [\beta+1,n], \]
i.e., all unseen branchings larger than~$\beta$ remain unchanged and valid on the stack~$U$.
All branchings smaller than~$\beta$, on the other hand, are unseen by definition, which justifies the call~$\branching([1,\beta-1])$ in step~P6.

We now explain how $\branching(I)$ computes the set~$B(x)\cap I$.
For reasons explained in the next section, the function does not compute~$B(x)$ directly, but rather the auxiliary set
\begin{equation}
\label{eq:XxI}
X(x,I):=\{y\in X-x\mid \lambda(x,y)\in I\}
\end{equation}
(step~B1).
This is the subset of all bitstrings~$y$ in~$X$ that are distinct from~$x$ and can be obtained from~$x$ by flipping a bit at a position inside the interval~$I$ and possibly some earlier bits.
In the suffix tree~$\cT(X)$, these are leaves that are reachable via a path that splits from the root-to-leaf path~$S(x)$ at a vertex whose distance from~$x$ is inside the interval~$I$.
We note that $X(x,I)=\emptyset$ if and only if $B(x)\cap I=\emptyset$.
Furthermore, for every $\beta\in B(x)\cap I$ there is a corresponding $y\in X(x,I)$ with $\lambda(x,y)=\beta$.
For any such~$y\in X(x,I)$ (chosen arbitrarily!) we can thus split the interval~$I\eqcolon [\alpha,\gamma]$ into two subintervals~$[\alpha,\beta-1]$ and~$[\beta+1,\gamma]$, which satisfy
\[ B(x)\cap [\alpha,\gamma]=(B(x)\cap [\alpha,\beta-1])\,\cup\,\{\beta\}\,\cup\, (B(x)\cap [\beta+1,\gamma]). \]
This is the recursive relation used in step~B2, where the recursive calls on the subintervals~$[\alpha,\beta-1]$ and~$[\beta+1,\gamma]$ are performed in reverse order, interspersed with pushing~$\beta$ onto the stack~$U$, to maintain the property that elements on the stack are sorted decreasingly.

Summarizing this discussion, we obtain the following result (cf.\ \cite[Theorem~18]{MR4795009}).

\begin{theorem}
\label{thm:algo}
Let $X\seq\{0,1\}^n$.
For every tiebreaking rule and every initial vertex~$\tx$, Algorithm~\upright{P\ppp{}} computes a genlex Hamilton path on the skeleton of~$\conv(X)$ starting at~$\tx$.
\end{theorem}

Figure~\ref{fig:protocol} shows how Algorithm~P\ppp{} computes the Hamilton path from Figure~\ref{fig:ssp}~(b).
We encourage the reader to work through this example in order to appreciate the inner workings of the algorithm.

\section{The bridge to combinatorial optimization}
\label{sec:opt}

From the pseudocode before we can extract the following two computational problems, which, if solved efficiently, directly lead to an efficient implementation of Algorithm~P\ppp{}:
\begin{enumerate}[leftmargin=8mm, noitemsep, topsep=1pt plus 1pt]
\item[\mybox{A\pppm{}}] Given a set $X\seq\{0,1\}^n$, an element~$x \in X$ and an interval~$I\seq [n]$, compute an element in the set~$X(x,I)$ defined in~\eqref{eq:XxI}, or decide that $X(x,I)=\emptyset$.
\item[\mybox{C}]  Given a set $X\seq\{0,1\}^n$, an element~$x \in X$ and an integer~$\beta \in [n]$ with
$N:=\argmin[y\in X-x \,\wedge\, \lambda(x,y)=\beta \mid d(x,y)]\neq\emptyset$,
compute an element in~$N$.
\end{enumerate}

Problems~A\pppm{} and~C can be reduced to instances of the following optimization problem, referred to as \defi{linear optimization with prescription}.
\begin{enumerate}[label=\mybox{LOP},leftmargin=12mm, noitemsep, topsep=1pt plus 1pt]
\item Given a set~$X\seq\{0,1\}^n$, a weight vector~$w\in W^n$ with $W\seq\mathbb{R}$, and disjoint sets $P_0,P_1\seq [n]$, compute an element in
$N:=\argmin[y\in X\,\wedge\, y_{P_0}=0 \,\wedge\, y_{P_1}=1 \mid w\cdot y]$,
or decide that this problem is infeasible, i.e., $N=\emptyset$.
\end{enumerate}
In this description we use $y_{P_b}=b$ for $b\in\{0,1\}$ as a shorthand notation for $y_i=b$ for all $i\in P_b$.
We refer to~$W$ as the \defi{weight set}.

We prove the following reduction from problem~A\pppm{} to problem~LOP.

\begin{lemma}
\label{lem:algoAp-LOP}
Suppose that problem~\upright{LOP} with weight set~$W=\{-1,0,+1\}$ can be solved in time~$t_{\upright{LOP}}=\Omega(n)$.
Then problem~\upright{A\pppm{}} can be solved in time~$\cO(t_{\upright{LOP}})$.
\end{lemma}

\begin{proof}
Consider a set~$X\seq\{0,1\}^n$, an element~$x\in X$ and an interval~$I\seq [n]$ as input for problem~A\pppm{}.
We define
\begin{equation*}
w_i:=\begin{cases}
-1 & \text{if } i\geq \min I\text{ and } x_i=0, \\
+1 & \text{if } i\geq \min I\text{ and } x_i=1, \\
0  & \text{if } i<\min I,
\end{cases}
\end{equation*}
and
\begin{equation*}
P_b:=\{i>\max I\mid x_i=b\} \text{ for } b\in\{0,1\}.
\end{equation*}
We claim that
\[\mu:=\min\limits_{y\in X\,\wedge\,y_{P_0}=0\,\wedge\,y_{P_1}=1} w\cdot y<w\cdot x=:a\]
if and only if~$X(x,I)\neq \emptyset$ holds for the set~$X(x,I)$ defined in~\eqref{eq:XxI}.
Indeed, if $\mu<a$, then there is a~$y^*\in X$ with $y^*_{P_0}=0$, $y^*_{P_1}=1$, and~$y^*_i\neq x_i$ for some $i\in I$.
It follows that $y^*\in X(x,I)$, which implies~$X(x,I)\neq \emptyset$.
Conversely, if~$X(x,I)\neq\emptyset$, then there is $y^*\in X-x$ with $\lambda(x,y^*)\in I$, i.e., there is a position $i\in I$ such that $y^*_i\neq x_i$ and $y^*_j=x_j$ for all $j>i$, in particular $y^*_{P_0}=0$ and $y^*_{P_1}=1$.
As disagreement from~$x$ on the interval~$I$ decreases the objective, we have $w\cdot y^*<w\cdot x$ and therefore~$\mu<a$.
This completes the proof of the lemma.
\end{proof}

A similar reduction from problem~C to problem~LOP was proved in~\cite{MR4795009}.

\begin{lemma}[{\cite[Lemma~21]{MR4795009}}]
\label{lem:algoC-LOP}
Suppose that problem~\upright{LOP} with weight set $W=\{-1,+1\}$ can be solved in time~$t_{\upright{LOP}}=\Omega(n)$.
Then problem~\upright{C} can be solved in time~$\cO(t_{\upright{LOP}})$.
\end{lemma}

\subsection{Main result}

The following is the main result of this paper, improving upon~\cite[Theorem~22]{MR4795009} by a $\log n$ factor.
The theorem asserts that efficiently solving prescription optimization on~$X$ yields an efficient algorithm for computing a Hamilton path on the skeleton of~$\conv(X)$.

\begin{theorem}
\label{thm:main}
Let $X\seq\{0,1\}^n$ and suppose that problem~\upright{LOP} with weight set $W=\{-1,0,+1\}$ can be solved in time~$t_{\upright{LOP}}=\Omega(n)$.
Then for every tiebreaking rule and every initial vertex~$\tx$, Algorithm~\upright{P\ppp{}} computes a genlex Hamilton path on the skeleton of~$\conv(X)$ starting at~$\tx$ with amortized delay~$\cO(t_{\upright{LOP}})$.
\end{theorem}

\begin{proof}
The correctness of the algorithm follows from Theorem~\ref{thm:algo}.
To complete the proof, by Lemmas~\ref{lem:algoAp-LOP} and~\ref{lem:algoC-LOP}, it suffices to argue that only a constant number of instances of problem~LOP must be solved per visited vertex on the Hamilton path.
This amortized analysis works as follows:
Lines~P4+P5 of Algorithm~P\ppp{} amount to solving one instance of problem~C.
Step~B1 of the function~$\branching()$ amounts to solving one instance of problem~A\pppm{}.
For the two recursive calls of $\branching()$ in line~B2, we account the two corresponding instances of problem~A\pppm{} to the operation~$U.\push(\beta)$.
For the call of $\branching()$ in line~P6, we account the one instance of problem~A\pppm{} to the preceding operation~$\beta\leftarrow U.\pop()$.
Similarly, the one instance of problem~C in steps~P4+P5 is accounted to the preceding~$U.\pop()$.
Thus, each item on the stack is accounted with solving three instances of problem~A\pppm{} and one instance of problem~C, i.e., four instances of problem~LOP.
As each item is popped from the stack eventually, and the total number of vertices visited is one more than the number of items on the stack, the total number of instances of the problem~LOP to be solved is exactly $4(|X|-1)+1\leq 4|X|$, where the $+1$ accounts for the one instance of problem~A\pppm{} solved in the initial call of $\branching()$ from line~P1.
\end{proof}

\subsection{Comparison to the algorithm from~\texorpdfstring{\cite{MR4795009}}{[MM24]}}

Algorithm~P\ppp{} follows the same underlying principles as Algorithm~P\sss{} from~\cite{MR4795009}, and it produces the same Hamilton path as output (assuming that the same starting vertex is chosen, and ties are broken consistently).
The main difference, which results in our improved delay, is the information maintained on the stack~$U$.
Recall that our Algorithm~P\ppp{} maintains on the stack~$U$ the set of all unseen branchings in decreasing order.
This is the most direct way to cast the algorithmic idea into a data structure.

Algorithm~P\sss{} from~\cite{MR4795009} also uses a stack, but it contains different information.
Specifically, it partitions the interval~$[1,n]$ into subintervals $I_1,\ldots,I_\ell$, and within each subinterval only maintains the \emph{smallest} unseen branching within that interval.
Thus, auxiliary problem~A from~\cite{MR4795009} asks to compute the smallest branching within an interval.
We replace this by our more lenient auxiliary problem~A\pppm{}, which does not ask for the smallest branching within an interval, but for an arbitrary branching within the interval.
Finding the smallest branching within an interval requires solving logarithmically many instances of problem~LOP, yielding the extra $\log n$ factor in the delay.
Only the information of the smallest branching discovered in the last step of the binary search is maintained, whereas other branchings discovered during the search are discarded.
In contrast, whenever our function $\branching()$ discovers an arbitrary branching~$\beta$, where $\beta=\lambda(x,y)$ for some $y\in X(x,I)$ (step~B1), it is pushed onto the stack~$U$ (step~B2), thus leading to a subsequent visit step.

\section{Computational study}
\label{sec:exp}

In this section we compare the Bussieck--L\"ubbecke~\cite{MR1659922} algorithm, called \texttt{zerone}, Algorithm~P\sss{} from~\cite{MR4795009} and our new Algorithm~P\ppp{} when solving the vertex enumeration problem for a variety of combinatorial 0/1-polytopes coming from set systems, graphs and posets.
We used the existing C/C++ implementations of the first algorithm and prepared implementations of the latter two algorithms in C++/Julia, interfacing them with the CPLEX linear programming (LP) solver.
For more details, see \cite{hladik:2026}.

\newcommand{\gap}{\rule{0pt}{3ex}}

\begin{table}[t!]
\caption{Results of computional experiments.
The bold numbers are the average delays of the three algorithms.}
\label{tab:exp}
\centerline{
\begin{tabular}{lrrrr>{\bf}rr>{\bf}rr>{\bf}r}
\toprule
$X$                                   & $n$ & $m$ & $|X|$ & \multicolumn{2}{l}{\texttt{zerone} \cite{MR1659922}} & \multicolumn{2}{l}{Alg.~P\sss{} \cite{MR4795009}} & \multicolumn{2}{l}{Alg.~P\ppp{}} \\
                                      &    &      &        & $T$ & $\nicefrac{T}{|X|}$ & $T$ & $\nicefrac{T}{|X|}$ & $T$ & $\nicefrac{T}{|X|}$ \\ \midrule
Uniform matroid $U_{10,5}$            & 10 &    2 &    252 &    1343 &   5.3 &    907 &   3.6 &     649 &   2.6 \\
Uniform matroid $U_{16,8}$            & 16 &    2 &  12870 &   71499 &   5.6 &  47831 &   3.7 &   33163 &   2.6 \\
Uniform matroid $U_{20,10}$           & 20 &    2 & 184756 & 1041351 &   5.6 & 695003 &   3.8 &  476388 &   2.6 \\
Uniform matroid $U_{70,1}$            & 70 &    2 &     70 &    4971 &  71.0 &    732 &  10.5 &     209 &   3.0 \\
Uniform matroid $U_{90,1}$            & 90 &    2 &     90 &    8191 &  91.0 &   1012 &  11.2 &     269 &   3.0 \\
\midrule
Vertex covers of $Q_4$                & 16 &   32 &    743 &    3623 &   4.9 &   2946 &   4.0 &    1519 &   2.0 \\
Vertex covers of $C_{14}\times P_2$   & 28 &   42 & 228487 & 1035551 &   4.5 & 890642 &   3.9 &  476528 &   2.1 \\
Vertex covers of $K_{10,10}$          & 20 &  100 &   2047 &   24553 &  12.0 &   8153 &   4.0 &    3071 &   1.5 \\
Vertex covers of $K_{14,14}$          & 28 &  196 &  32767 &  524257 &  16.0 & 130309 &   4.0 &   49151 &   1.5 \\
Matchings of $P_4\times P_4$ \gap     & 24 &   16 &  10012 &   50773 &   5.1 &  40897 &   4.1 &   19472 &   1.9 \\
Matchings of $C_8\times P_2$          & 24 &   16 &  11396 &   70153 &   6.2 &  46668 &   4.1 &   23768 &   2.1 \\
Matchings of $K_9$                    & 36 &  256 &   2620 &   35785 &  13.7 &  15270 &   5.8 &    7209 &   2.8 \\
Matchings of $W_{14}$                 & 26 & 7360 &   3550 &   15827 &   4.5 &  19948 &   5.6 &    6875 &   1.9 \\
Perfect matchings of $K_{6,6}$ \gap   & 36 &   11 &    720 &   19859 &  27.6 &   5045 &   7.0 &    2581 &   3.6 \\
Perfect matchings of $K_{8,8}$        & 64 &   15 &  40320 & 1492959 &  37.0 & 308847 &   7.7 &  145138 &   3.6 \\
Perfect matchings of $P_6\times P_6$  & 60 &   35 &   6728 &  170179 &  25.3 &  49458 &   7.4 &   24799 &   3.7 \\
Perfect matchings of $C_7\times P_2$  & 21 & 7477 &     29 &     729 &  25.1 &    172 &   5.9 &     104 &   3.6 \\
Perfect matchings of $K_{10}$         & 45 &  511 &    945 &   32463 &  34.4 &   7284 &   7.7 &    3092 &   3.3 \\
Perfect matchings of $W_{18}$         & 34 & 124118 &   17 &     885 &  52.1 &    138 &   8.1 &      50 &   2.9 \\
Spanning trees of $K_7$ \gap          & 21 &   99 &  16807 &  137189 &   8.2 &  71038 &   4.2 &   53170 &   3.2 \\
Spanning trees of $K_{4,4}$           & 16 &  121 &   4096 &   26585 &   6.5 &  16371 &   4.0 &   12272 &   3.0 \\
Spanning trees of $Q_3$               & 12 &   67 &    384 &    2383 &   6.2 &   1499 &   3.9 &    1152 &   3.0 \\
Spanning trees of $W_{11}$            & 20 &  902 &  15125 &  131431 &   8.7 &  69907 &   4.6 &   48875 &   3.2 \\
\midrule
Antichains of $L_{50}$                & 50 &   49 &     51 &    2551 &  50.0 &    267 &   5.2 &     100 &   2.0 \\
Antichains of $L_{90}$                & 90 &   89 &     91 &    8191 &  90.0 &    547 &   6.0 &     180 &   2.0 \\
Antichains of $Q_5$                   & 32 &   80 &   7581 &  103923 &  13.7 &  40650 &   5.4 &   20358 &   2.7 \\
Antichains of $\Cr_{12}$              & 24 &   24 & 103682 &  619237 &   6.0 & 394382 &   3.8 &  173520 &   1.7 \\
Antichains of $S_{14}$                & 28 &  182 &  32781 &  524439 &  16.0 & 130399 &   4.0 &   49189 &   1.5 \\
Antichains of $\Pi_4$                 & 24 &   36 &    856 &   11335 &  13.2 &   4304 &   5.0 &    2159 &   2.5 \\
Ideals of $L_{50}$ \gap               & 50 &    1 &     51 &    2551 & 50.0 &     486 &   9.5 &     149 &   2.9 \\
Ideals of $L_{90}$                    & 90 &    1 &     91 &    8191 & 90.0 &    1020 &  11.2 &     269 &   3.0 \\
Ideals of $Q_5$                       & 32 &  120 &   7581 &  103923 & 13.7 &   41298 &   5.4 &   20491 &   2.7 \\
Ideals of $\Cr_{12}$                  & 24 &   24 & 103682 &  619237 &  6.0 &  395512 &   3.8 &  173521 &   1.7 \\
Ideals of $S_{14}$                    & 28 &  182 &  32781 &  524439 & 16.0 &  130423 &   4.0 &   49191 &   1.5 \\
Ideals of $\Pi_4$                     & 24 &   16 &    856 &   11335 & 13.2 &    4393 &   5.1 &    2229 &   2.6 \\
\bottomrule
\end{tabular}}
\end{table}

\pgfdeclarelayer{background}    
\pgfsetlayers{background,main}
\definecolor{lightred}{rgb}{1.0,0.9,0.9}

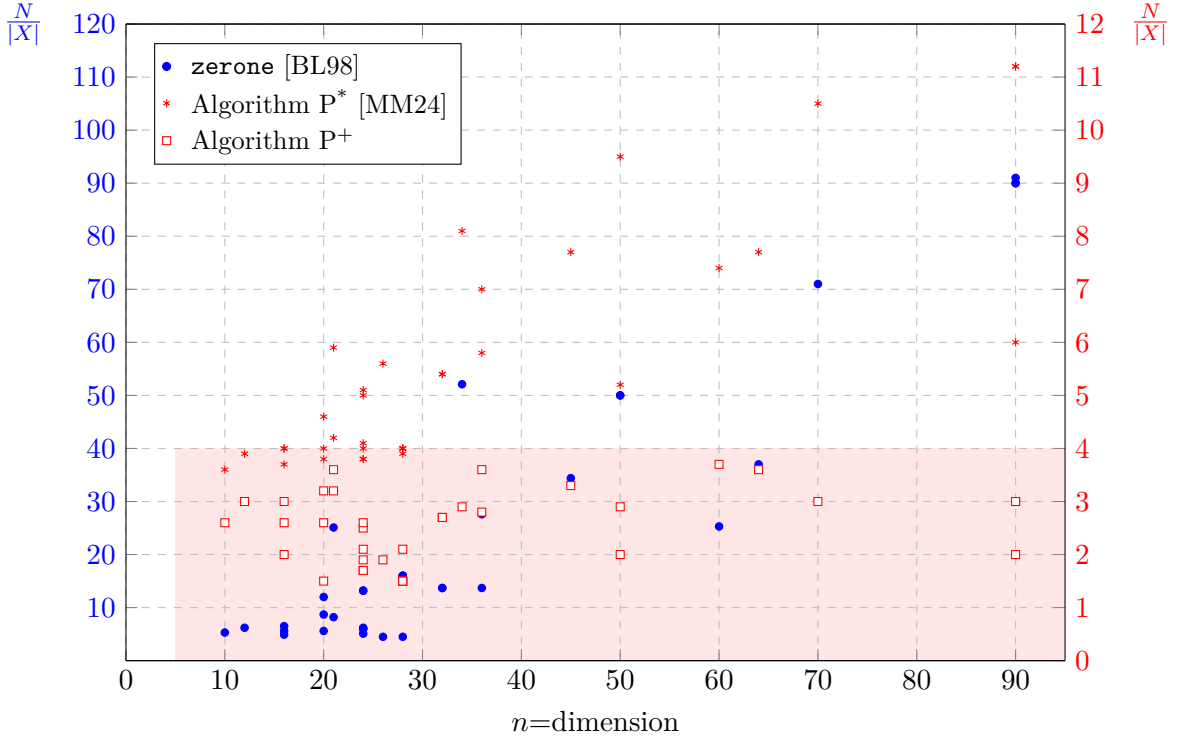
\begin{figure}[b!]
\begin{tikzpicture}
\begin{axis}[
    axis y line*=left,
    xlabel={$n$=dimension},
    ylabel={$\frac{T}{|X|}$},
    ylabel style={color=blue,rotate=-90,at={(-0.01,1.0)}},
    xmin=0, xmax=95,
    ymin=0, ymax=120,
    xtick={0,10,20,30,40,50,60,70,80,90},
    yticklabel style={color=blue},
    ytick={10,20,30,40,50,60,70,80,90,100,110,120},
    ymajorgrids=true,
    xmajorgrids=true,
    grid style=dashed,
    width=14cm,
    height=10cm,
]
\addplot[color=blue,mark=*,mark size=1.4pt,only marks] coordinates {(10,5.3)(16,5.6)(20,5.6)(70,71.0)(90,91.0)(16,4.9)(28,4.5)(20,12.0)(28,16.0)(24,5.1)(24,6.2)(36,13.7)(26,4.5)(36,27.6)(64,37.0)(60,25.3)(21,25.1)(45,34.4)(34,52.1)(21,8.2)(16,6.5)(12,6.2)(20,8.7)(50,50.0)(90,90.0)(32,13.7)(24,6.0)(28,16.0)(24,13.2)(50,50.0)(90,90.0)(32,13.7)(24,6.0)(28,16.0)(24,13.2)};
\begin{pgfonlayer}{background}
\fill[lightred] (50,0) rectangle (950,40);
\end{pgfonlayer}
\end{axis}
\begin{axis}[
    legend style={font=\small},
    legend pos=north west,
    legend cell align={left},
    axis y line*=right,
    ylabel={$\frac{T}{|X|}$},
    ylabel style={color=red,rotate=-90,at={(1.19,1.0)}},
    xmin=0, xmax=95,
    ymin=0, ymax=12,
    xmajorticks=false,
    yticklabel style={color=red},
    ytick={0,1,2,3,4,5,6,7,8,9,10,11,12},
    width=14cm,
    height=10cm,
]
\addlegendimage{mark=*,blue,only marks,mark size=1.5}
\addlegendimage{mark=asterisk,red,only marks,mark size=1.5}
\addlegendimage{mark=square,red,only marks,mark size=1.5}
\legend{\hspace{2mm}\texttt{zerone} \cite{MR1659922},\hspace{2mm}Algorithm~P\sss{} \cite{MR4795009},\hspace{2mm}Algorithm~P\ppp{}}
\addplot[color=red,mark=asterisk,mark size=1.5pt,only marks] coordinates {(10,3.6)(16,3.7)(20,3.8)(70,10.5)(90,11.2)(16,4.0)(28,3.9)(20,4.0)(28,4.0)(24,4.0)(24,4.1)(36,5.8)(26,5.6)(36,7.0)(64,7.7)(60,7.4)(21,5.9)(45,7.7)(34,8.1)(21,4.2)(16,4.0)(12,3.9)(20,4.6)(50,5.2)(90,6.0)(32,5.4)(24,3.8)(28,4.0)(24,5.0)(50,9.5)(90,11.2)(32,5.4)(24,3.8)(28,4.0)(24,5.1)};
\addplot[color=red,mark=square*,mark options={fill=white},mark size=1.5pt,only marks] coordinates {(10,2.6)(16,2.6)(20,2.6)(70,3.0)(90,3.0)(16,2.0)(28,2.1)(20,1.5)(28,1.5)(24,1.9)(24,2.1)(36,2.8)(26,1.9)(36,3.6)(64,3.6)(60,3.7)(21,3.6)(45,3.3)(34,2.9)(21,3.2)(16,3.0)(12,3.0)(20,3.2)(50,2.0)(90,2.0)(32,2.7)(24,1.7)(28,1.5)(24,2.5)(50,2.9)(90,3.0)(32,2.7)(24,1.7)(28,1.5)(24,2.6)};
\end{axis}
\end{tikzpicture}
\caption{Plot of average delays of the three algorithms for all polytopes listed in Table~\ref{tab:exp}.
Note that two different vertical axes are used, as the values for the first algorithm (blue, left axis) are by an order of magnitude (=factor~10) larger than the values for the latter two algorithms (red, right axis).
The shaded region highlights values $\leq 4$, and all values for Algorithm~P\ppp{} are contained in it.}
\label{fig:exp}
\end{figure}

Each polytope is encoded by its hyperplane description $P\coloneq \{x\in\mathbb{R}^n\mid Ax\leq b\}$, where $A\in\mathbb{R}^{m\times n}$ and $b\in\mathbb{R}^m$, and each algorithm reports all elements of the vertex set $X\seq\{0,1\}^n$ of~$P$.
For the latter two algorithms the vertices in~$X$ are listed in the order of a genlex Hamilton path on the polytope~$P$.
We write~$T$ for the number of instances of the problem~\upright{LOP} solved within each algorithm (recall Theorem~\ref{thm:main}), and we use the quantity~$T/|X|$ as a proxy for the average delay.
Of course, in practice, the running time not only depends on the number of optimization instances solved, but also on their size and structure, and on the LP solver being used, but this is something we want to abstract from.
Furthermore, in practice I/O considerations, namely quickly passing data to and from the LP solver, are also of major relevance (see~\cite{hladik:2026} for details).

We now describe the combinatorial objects and corresponding 0/1-polytopes considered for our experiments.

For set systems, we consider the base polytope of the uniform matroid~$U_{n,k}$, which has as bases all $k$-element subsets of an $n$-element ground set.
The corresponding hyperplane description is simply $\sum_{i\in[n]}x_i=k$.

Furthermore, we consider vertex covers, matchings, perfect matchings, and spanning trees in graphs, with the corresponding polytopes being the vertex cover polytope~\cite{MR510371}, matching polytope, perfect matching polytope \cite{MR371732} and spanning tree polytope~\cite{MR510371}.
Their hyperplane descriptions are well-known, see the aforementioned papers and~\cite{hladik:2026} for details.
The following notations are used for the different families of graphs.
$K_n$ denotes the complete graph on $n$ vertices.
$K_{a,b}$ denotes the complete bipartite graph with partition classes of size~$a$ and~$b$.
$Q_n$ is the $n$-dimensional hypercube.
$W_n$ is the wheel with $n$ vertices in total, $n-1$ of them with degree~3 and one of them with degree~$n-1$.
$P_n$ and $C_n$ denote the path and cycle on $n$ vertices, respectively.
$G\times H$ denotes the Cartesian product of two graphs~$G$ and~$H$.

Lastly, we consider antichains and ideals (downward closed sets) in posets, with the corresponding polytopes being the order polytope and chain polytope~\cite{MR824105}.
The following notations are used for the different families of posets.
$L_n$ denotes the total order (=chain) on $n$ elements.
$Q_n$ denotes the Boolean lattice of subsets of~$[n]$.
$\Cr_n$ and $S_n$ denote the crown poset and standard example with $2n$ elements, respectively.
$\Pi_n$ denotes the weak Bruhat order on permutations of~$[n]$.

Recall from the proof of Theorem~\ref{thm:main} that if we implement Algorithm~P\ppp{} exactly as described in Section~\ref{sec:algo}, then the resulting number~$T$ will be exactly~$T=4(|X|-1)+1=4|X|-3$, yielding an average delay of exactly~$T/|X|=4-3/|X|<4$.
From the example shown in Figure~\ref{fig:protocol} we notice a straightforward improvement to the function~$\branching(I)$.
Namely, if the interval~$I$ is empty, i.e., $I=\emptyset$, then the condition $X(x,I)=\emptyset$ in line~B1 is trivially satisfied, so we can return immediately without changing the stack, thus saving one call to the LP solver.
The results reported about Algorithm~P\ppp{} in Table~\ref{tab:exp} use this slightly improved version of the algorithm, demonstrating in particular how much can be gained compared to the average delay~$4-3/|X|$ that we would obtain without this improvement.

\subsection{Evaluation}

Figure~\ref{fig:exp} shows the average delays of all algorithms (proxied by the quantity $T/|X|$) plotted as a function of~$n$, the number of variables of the problem, for all polytopes listed in Table~\ref{tab:exp}.
The average delays of the algorithm \texttt{zerone} are roughly an order of magnitude larger than those of the other two algorithms, which we accomodate by using two different vertical axis labelings.
Furthermore, these delays grow with~$n$, possibly linearly, as expected.
The average delays of Algorithm~P\ppp{} are always strictly smaller than those of Algorithm~P\sss{}, often by a factor of 2--3.
For Algorithm~P\sss{} they grow with~$n$, possibly logarithmically, as expected, whereas for Algorithm~P\ppp{} they remain constant and strictly smaller than~4, often considerably smaller than~4, as predicted by Theorem~\ref{thm:main}.

\bibliographystyle{alpha}
\bibliography{refs}

\end{document}